\documentclass[submission,copyright,creativecommons]{eptcs}
\providecommand{\event}{QPL 2017}
\usepackage{amsthm,amsmath,amssymb,mathtools}
\usepackage{tikz,tikz-cd}
\usepackage{color}
\usepackage{enumitem} 

\pgfdeclarelayer{foreground}
\pgfdeclarelayer{background}
\pgfdeclarelayer{morphismlayer}
\pgfsetlayers{background,main,morphismlayer,foreground}

\usetikzlibrary{matrix}
\usetikzlibrary{decorations.markings}
\usetikzlibrary{shapes.geometric}
\usetikzlibrary{arrows}

\tikzstyle{morphism}=[font=\small,morphismshape]
\tikzstyle{box}=[rectangle,inner sep=.4ex, draw=black, node on layer=foreground]

\tikzstyle{groundtwo} = [semicircle, draw=black, fill=white,scale=0.8, node on layer=foreground]
\tikzstyle{groundthree} = [semicircle, draw=black, fill=black,scale=0.8, node on layer=foreground]

\newif\ifvflip\pgfkeys{/tikz/vflip/.is if=vflip}
\newif\ifhflip\pgfkeys{/tikz/hflip/.is if=hflip}
\newif\ifhvflip\pgfkeys{/tikz/hvflip/.is if=hvflip}
\newlength\morphismheight
\newlength\wedgewidth
\tikzset{width/.initial=1mm}

\newlength\minimummorphismwidth
\newlength\stateheight
\newlength\minimumstatewidth
\newlength\connectheight
\tikzset{colour/.initial=white}

\tikzstyle{mixed}=[line width=.7pt]
\tikzstyle{pure}=[line width=.7pt]
\tikzset{diredge/.style={decoration={
  markings,
  mark=at position 0.525 with {\arrow{#1}}},postaction={decorate}}}
\tikzset{
    diredge/.default=>
}

\tikzset{diredgestart/.style={decoration={
  markings,
  mark=at position 4pt with {\arrow{#1}}},postaction={decorate}}}
\tikzset{
    diredgestart/.default=<
}

\tikzset{diredgeend/.style={decoration={
  markings,
  mark=at position 1 with {\arrow{#1}}},postaction={decorate}}}
\tikzset{
    diredgeend/.default=>
}

\makeatletter

\pgfdeclareshape{morphismshape}
{
    \savedanchor\centerpoint
    {
        \pgf@x=0pt
        \pgf@y=0pt
    }
    \anchor{center}{\centerpoint}
    \anchorborder{\centerpoint}
    \saveddimen\overallwidth
    {
        \pgfkeysgetvalue{/tikz/width}{\minwidth}
        \pgf@x=\wd\pgfnodeparttextbox
        \ifdim\pgf@x<\minwidth
            \pgf@x=\minwidth
        \fi
    }
    \savedanchor{\upperrightcorner}
    {
        \pgf@y=.5\ht\pgfnodeparttextbox
        \advance\pgf@y by -.5\dp\pgfnodeparttextbox
        \pgf@x=.5\wd\pgfnodeparttextbox
    }
    \anchor{north}
    {
        \pgf@x=0pt
        \pgf@y=0.5\morphismheight
    }
    \anchor{north east}
    {
        \pgf@x=\overallwidth
        \multiply \pgf@x by 4
        \divide \pgf@x by 5
        \pgf@y=0.5\morphismheight
    }
    \anchor{east}
    {
        \pgf@x=\overallwidth
        \divide \pgf@x by 2
        \advance \pgf@x by 5pt
        \pgf@y=0pt
    }
    \anchor{west}
    {
        \pgf@x=-\overallwidth
        \divide \pgf@x by 2
        \advance \pgf@x by -5pt
        \pgf@y=0pt
    }
    \anchor{north west}
    {
        \pgf@x=-\overallwidth
        \multiply \pgf@x by 4
        \divide \pgf@x by 5
        \pgf@y=0.5\morphismheight
    }
    \anchor{south east}
    {
        \pgf@x=\overallwidth
        \multiply \pgf@x by 4
        \divide \pgf@x by 5
        \pgf@y=-0.5\morphismheight
    }
    \anchor{south west}
    {
        \pgf@x=-\overallwidth
       \multiply \pgf@x by 4
        \divide \pgf@x by 5
        \pgf@y=-0.5\morphismheight
    }
    \anchor{south}
    {
        \pgf@x=0pt
        \pgf@y=-0.5\morphismheight
    }
    \anchor{text}
    {
        \upperrightcorner
        \pgf@x=-\pgf@x
        \pgf@y=-\pgf@y
    }
    \backgroundpath
    {
    \begin{scope}
        \pgfkeysgetvalue{/tikz/fill}{\morphismfill}
        \pgfsetstrokecolor{black}
        \begin{scope}
        \pgfsetstrokecolor{black}
        \pgfsetfillcolor{white}
                \ifhflip
                    \pgftransformyscale{-1}
                \fi
                \ifvflip
                    \pgftransformxscale{-1}
                \fi
                \ifhvflip
                    \pgftransformxscale{-1}
                    \pgftransformyscale{-1}
                \fi
                \pgfpathmoveto{\pgfpoint
                    {-0.5*\overallwidth-5pt}
                    {0.5*\morphismheight}}
                \pgfpathlineto{\pgfpoint
                    {0.5*\overallwidth+5pt}
                    {0.5*\morphismheight}}
                \pgfpathlineto{\pgfpoint
                    {0.5*\overallwidth + \wedgewidth}
                    {-0.5*\morphismheight}}
                \pgfpathlineto{\pgfpoint
                    {-0.5*\overallwidth-5pt}
                    {-0.5*\morphismheight}}
                \pgfpathclose
                \pgfusepath{fill,stroke}
        \end{scope}
    \end{scope}
    }
}

\pgfdeclareshape{ground}
{
    \savedanchor\centerpoint
    {
        \pgf@x=0pt
        \pgf@y=0pt
    }
    \anchor{center}{\centerpoint}
    \anchorborder{\centerpoint}

    \anchor{north}
    {
        \pgf@x=0pt
        \pgf@y=0.16\stateheight
    }
    \anchor{south}
    {
        \pgf@x=0pt
        \pgf@y=0pt
    }
    \saveddimen\overallwidth
    {
        \pgfkeysgetvalue{/pgf/minimum width}{\minwidth}
        \pgf@x=\minimumstatewidth
        \ifdim\pgf@x<\minwidth
            \pgf@x=\minwidth
        \fi
    }
    \backgroundpath
    {
        \begin{pgfonlayer}{foreground}
        \pgfsetstrokecolor{black}
        \pgfsetlinewidth{1.25pt}
        \ifhflip
            \pgftransformyscale{-1}
        \fi
        \pgftransformscale{0.5}
        \pgfpathmoveto{\pgfpoint{-0.5*\overallwidth}{0pt}}
        \pgfpathlineto{\pgfpoint{0.5*\overallwidth}{0pt}}
        \pgfpathmoveto{\pgfpoint{-0.33*\overallwidth}{0.33*\stateheight}}
        \pgfpathlineto{\pgfpoint{0.33*\overallwidth}{0.33*\stateheight}}
        \pgfpathmoveto{\pgfpoint{-0.16*\overallwidth}{0.66*\stateheight}}
        \pgfpathlineto{\pgfpoint{0.16*\overallwidth}{0.66*\stateheight}}
        \pgfpathmoveto{\pgfpoint{-0.02*\overallwidth}{\stateheight}}
        \pgfpathlineto{\pgfpoint{0.02*\overallwidth}{\stateheight}}
        \pgfusepath{stroke}
        \end{pgfonlayer}
    }
}

\makeatletter
\pgfkeys{%
  /tikz/on layer/.code={
    \pgfonlayer{#1}\begingroup
    \aftergroup\endpgfonlayer
    \aftergroup\endgroup
  },
  /tikz/node on layer/.code={
    \gdef\node@@on@layer{%
      \setbox\tikz@tempbox=\hbox\bgroup\pgfonlayer{#1}\unhbox\tikz@tempbox\endpgfonlayer\egroup}
    \aftergroup\node@on@layer
  },
  /tikz/end node on layer/.code={
    \endpgfonlayer\endgroup\endgroup
  }
}
\def\node@on@layer{\aftergroup\node@@on@layer}
\makeatother

\makeatother

\newcommand{\lift}[1][1]{\mathop{
\smash{\raisebox{0pt}{\hspace{-1pt}\ensuremath{\begin{pic}[scale=0.2*#1]
   \draw (0,0) to (1,0) to (1,1) to (0,1) to (0,0) to (1,1);
\end{pic}
}}}}}

\newcommand{\tinyground}[1][ground]{
\smash{\raisebox{-2pt}{\hspace{-3pt}\ensuremath{\begin{pic}[scale=0.4]
    \node[#1, scale=0.6] (1) at (0,0.4) {};
    \draw [pure] (1.south) to +(0,-.3);
\end{pic}
}}}}

\newcommand{\tinymix}[1][ground]{
\smash{\raisebox{-2pt}{\hspace{-3pt}\ensuremath{\begin{pic}[scale=-0.4]
    \node[#1, scale=-0.6] (1) at (0,0.4) {};
    \draw [pure] (1.south) to +(0,-.3);
\end{pic}
}}}}

\newenvironment{pic}[1][]
{\begin{aligned}\begin{tikzpicture}[font=\tiny,#1]}
{\end{tikzpicture}\end{aligned}}

\newcommand{\cat}[1]{\ensuremath{\mathbf{#1}}}
\newcommand{\id}[1][]{\ensuremath{\mathrm{id}_{#1}}}

\newcommand{\llift}[1]{\ensuremath{^{\lift[.7]}\!\!{#1}}}
\newcommand{\rlift}[1]{\ensuremath{{#1}^{\,\lift[.7]}}}
\newcommand{\Llift}[1]{\ensuremath{^{\oslash}\!{#1}}}
\newcommand{\Rlift}[1]{\ensuremath{{#1}^{\oslash}}}

\theoremstyle{plain}
\newtheorem{theorem}{Theorem}
\newtheorem{lemma}[theorem]{Lemma}

\theoremstyle{definition}
\newtheorem{definition}[theorem]{Definition}
\newtheorem{example}[theorem]{Example}

\title{Purity through Factorisation}
\author{Oscar Cunningham
    \institute{University of Oxford}
    \thanks{Supported by EPSRC Studentship OUCL/2014/OAC.}
    \email{oscar.cunningham@cs.ox.ac.uk}
  \and Chris Heunen 
    \institute{University of Edinburgh}
    \thanks{Supported by EPSRC Fellowship EP/L002388/1.}
    \email{chris.heunen@ed.ac.uk}
}
\def\titlerunning{Purity through Factorisation}
\def\authorrunning{O. Cunningham \& C. Heunen}

\begin{document}
\maketitle
\begin{abstract}
  We give a construction that identifies the collection of pure processes (\textit{i.e.}\ those which are deterministic, or without randomness) within a theory containing both pure and mixed processes. Working in the framework of symmetric monoidal categories, we define a  pure subcategory. This definition arises elegantly from the categorical notion of a weak factorisation system. Our construction gives the expected result in several examples, both quantum and classical. 
\end{abstract}

\section{Introduction}

Categorical quantum mechanics models physical theories as symmetric monoidal categories: objects are interpreted as physical systems, and morphisms are interpreted as processes that take a state of one system to a state of another~\cite{abramskycoecke:categoricalsemantics,coeckekissinger:cqm,heunenvicary:cqm}.
This approach captures various physical theories uniformly:
\begin{itemize}
\item functions between finite sets, which may be interpreted as deterministic processes between systems with finitely many states;
\item relations between finite sets, which may be interpreted as nondeterministic processes between systems with finitely many states;
\item stochastic matrices (those matrices with entries in the nonnegative reals, whose columns sum to 1), which may be interpreted as probabilistic processes between systems with finitely many states;
\item completely positive maps between finite-dimensional Hilbert spaces, which may be interpreted as quantum processes between systems with finitely many degrees of freedom.
\end{itemize}
In these interpretations, some, but not all, processes are \emph{pure}, in the sense that probability plays no role:
\begin{itemize}
\item a function is always pure;
\item a relation is pure when it is (the graph of) a partial function;
\item a stochastic matrix is pure when each column has one nonzero entry;
\item a completely positive map is pure when it needs no ancilla.
\end{itemize}
General morphisms are interpreted as a \emph{mixture} of pure ones. The goal of this article is to give an \emph{operational definition} of purity.
In contrast to other approaches~\cite{chiribella:pure,chiribellascandolo:diagonal,chiribella:dilation,chiribelladarianoperinotti:purification,coeckeperdrix:environment,cunninghamheunen:cpstar}, our definition merely needs the structure of symmetric monoidal categories. That is, we define what it means for a morphism to be pure using only its relationships with the other morphisms in the category (so without daggers or dual objects), and without any reference to the interpretation of morphisms as being processes between systems. 
The key idea is that \emph{purification} may be regarded as part of a \emph{factorisation system}.
To evidence that our definition is useful, we prove that in the above four example cases it recaptures the desired interpretation.

We start by recalling factorisation systems in Section~\ref{sec:factorisation}, and discuss lifting properties in Section~\ref{sec:lifting}. Section~\ref{sec:purity} then defines purification in terms of factorisation systems, and Section~\ref{sec:examples} illustrates that this is a good notion by showing that it captures probabilistic theories, possibilistic theories, deterministic theories, and quantum theories. Finally, Section~\ref{sec:related} gives a comparison with previous attempts at operational definitions of purity~\cite{chiribella:pure,selbycoecke:leaks}. We will use the diagrammatic notation~\cite{selinger:graphicallanguages}.

\section{Factorisation}\label{sec:factorisation}

Our definition of purity is inspired by Stinespring's dilation theorem~\cite{stinespring}. This is a theorem about the category of finite dimensional Hilbert spaces and completely positive maps~\cite{choi}; Stinespring's statement allowed infinite dimensional spaces, but we don't use that generality. Hilbert spaces represent quantum systems and completely positive maps represent all the physically realisable nondeterministic processes between them~\cite{selinger:cpm}. The pure morphisms representing deterministic processes form a symmetric monoidal subcategory. Additionally, every system has a has a \emph{completely mixed state} represented by a completely positive map from $\mathbb C$ (the monoidal unit) to the Hilbert space representing that system~\cite{coeckeperdrix:environment,cunninghamheunen:cpstar}. We will draw this state diagrammatically as \tinymix.

\begin{theorem}[Stinespring dilation]\label{thm:stinespring}
  Any completely positive map $A\to B$ can be written in the form
  \begin{equation}\label{eq:stinespring}
    \begin{pic}[xscale=.75,yscale=.5]
      \node[morphism] (a) at (0,.5) {$\;p\;$};
      \draw (a.north) to +(0,.5)node[above]{B};
      \draw (a.south west) to +(0,-.9)node[below]{A};
      \draw (a.south east) to +(0,-.3)node[right]{C} to +(0,-.6)node[ground,scale=-.75]{};
    \end{pic}
  \end{equation}
  for some \emph{ancilla} Hilbert space $C$ and pure morphism $p$. 
\end{theorem}
\begin{proof}
  See~\cite{selinger:cpm}.
\end{proof}
  
We can reformulate the Stinespring dilation theorem (and theorems like it) in a category-theoretic way by using a version of factorisation systems~\cite{freydkelly:factorisation,adamekherrlichstrecker:joyofcats}. 

\begin{definition}
  Suppose \cat{C} is a category and $\mathcal L$ and $\mathcal R$ are collections of morphisms in \cat{C}. We say that $(\mathcal L,\mathcal R)$ is a \emph{factorisation system} if every morphism in \cat{C} may be written as $r\circ l$ for some $l\in \mathcal L$ and $r\in \mathcal R$. (We do not necessarily demand that $\mathcal L$ and $\mathcal R$ are subcategories, or that they have any other properties at all.)
\end{definition}

Stinespring's theorem says that the category of completely positive maps has a factorisation system, with $\mathcal L$ consisting of all maps that introduce a mixed state, and $\mathcal R$ consisting of all pure maps. See also~\cite{westerbaan:paschke}.
Two types of factorisation system are of particular interest:
\begin{itemize}
  \item \emph{Orthogonal factorisation systems} are typified by the system $(\mathrm{Surj}, \mathrm{Inj})$ in \cat{Set}; the factorisation is found by writing a function as a surjection onto its image followed by inclusion of the image into the codomain. Orthogonal factorisation systems have the property that their factorisations are unique up to a unique isomorphism on the `middle' object through which the morphism is being factored. In fact orthogonal factorisation systems can be defined as those which have this uniqueness property in addition to the property that $\mathcal L$ and $\mathcal R$ are replete\footnote{A subcategory is \emph{replete} if it contains (all objects and) all isomorphisms.} subcategories~\cite[C.0.19]{joyal:quasicategories}.

  \item \emph{Weak factorisation systems}\footnote{To avoid confusion: being a weak factorisation system isn't weaker than being a factorisation system. It is called that merely because being weak is weaker than being orthogonal.} are typified by the system $(\mathrm{Inj}, \mathrm{Surj})$ in \cat{Set}; the factorisation is found by letting the `middle' object be the domain along with one extra element for each point not in the image of the original function.  In a weak factorisation system the factorisations are not necessarily unique, but they do satisfy a weaker property which we will give in the next section.
\end{itemize}

The Stinespring dilation is known to never be unique. So it cannot describe an orthogonal factorisation system, but might describe a weak factorisation system.  In fact, we will see that it does indeed give a weak factorisation system after a slight modification, namely expanding the collection $\mathcal L$ to contain all isomorphisms. This might be expected from the interpretation: the introduction of an ancillary completely mixed state resembles an injection as it maps a smaller space into a larger one, so we might have predicted in advance that there would be an analogy with $(\mathrm{Inj}, \mathrm{Surj})$ but not with $(\mathrm{Surj}, \mathrm{Inj})$.

\section{Lifting}\label{sec:lifting}

Weak factorisation systems are defined in terms of the following relation. Let \cat{C} be any category and let $f$ and $g$ be morphisms in \cat{C} (there need not be any relation between their domains and codomains). Say that $f$ and $g$ have the \emph{lifting property}, and write $f \lift g$, when any commuting square
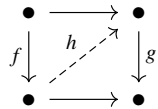
\[
\begin{tikzcd}
  \bullet\arrow{r}\arrow{d}[swap]{f}&\bullet\arrow{d}{g}\\
  \bullet\arrow{r}\arrow[dashed]{ur}{h}&\bullet
\end{tikzcd}
\]
has a morphism $h$ making the two triangles commute.

For a collection $\mathcal{A}$ of morphisms in \cat{C} (that are not necessarily closed under composition), define $\rlift{\mathcal A}$ and $\llift{\mathcal A}$ to be the collections of morphisms that have the lifting property on the right or left of all the morphisms in $\mathcal{A}$:
\begin{align*}
  \rlift{\mathcal A} & = \{g\in\mathrm{Mor}(\cat C)|\forall f\in\mathcal A:f\lift g\}\text, \\
  \llift{\mathcal A} & = \{f\in\mathrm{Mor}(\cat C)|\forall g\in\mathcal A:f\lift g\}\text.
\end{align*}

It is easily shown that $\rlift{\mathcal A}$ and $\llift{\mathcal A}$ contain all isomorphisms and are closed under composition, and are therefore replete subcategories of \cat{C}. We also have that $\llift{(-)}$ and $\rlift{(-)}$ are order reversing in the sense that if $\mathcal A\subseteq\mathcal B$ then $\rlift{\mathcal A} \supseteq \rlift{\mathcal B}$ and $\llift{\mathcal{A}} \supseteq\,\, \llift{\mathcal B}$. Furthermore $\mathcal{A} \subseteq \rlift{(\llift{\mathcal A})}$ and $\mathcal A \subseteq \llift{(\rlift{\mathcal A})}$; pairs of order reversing functions with this property are known as \emph{Galois connections}. It follows that alternating applications of $\llift{(-)}$ and $\rlift{(-)}$ eventually cease to have an effect: $\llift{(\rlift{(\llift{\mathcal A})})} = \llift{\mathcal A}$ and $\rlift{(\llift{(\rlift{\mathcal A})})} = \rlift{\mathcal A}$.

Given any collection $\mathcal A$ of morphisms, we can generate a pair $(\mathcal L,\mathcal R)$ by letting $\mathcal R = \rlift{\mathcal A}$ and then $\mathcal L=\,\llift{\mathcal R}$. This pair then has the property that $\mathcal L={}^{\lift[0.7]}\mathcal R$ and $\mathcal R=\rlift{\mathcal L}$. \footnote{In the theory of model categories, $(\mathcal L,\mathcal R)$ is called \emph{cofibrantly generated} from $\mathcal A$. We won't use this terminology here.} Thus for any collection of morphisms $\mathcal A$ we can generate a pair $(\mathcal L,\mathcal R)$ and then ask if this pair happens to form a factorisation system.

\begin{definition}
  A \emph{weak factorisation system} is a factorisation system with $\mathcal L=\,\llift{\mathcal R}$ and $\mathcal R=\rlift{\mathcal L}$.
\end{definition}

\begin{example}\label{ex:injsurjweakfactorisationsystem}
  Let us see why $(\mathrm{Inj}, \mathrm{Surj})$ is a weak factorisation system in \cat{Set}.

  First note that this is indeed a factorisation system. Given a function $f \colon A\to B$ let $i$ be the inclusion $A\to A\sqcup\left(B\smallsetminus\mathrm{im}(f)\right)$ and let $s=f \sqcup \id[B\smallsetminus\mathrm{im}(f)] \colon A\sqcup\left(B\smallsetminus\mathrm{im}(f)\right)\to \mathrm{im}(f)\sqcup\left(B\smallsetminus\mathrm{im}(f)\right)=B$. Then $i$ is injective, $s$ is surjective, and $s\circ i = f$.

  Next we want to show that $i\lift s$ for any injection $i\colon A\to B$ and surjection $s\colon C\to D$. So suppose that
  \[
   \begin{tikzcd}
    A\arrow{r}{j}\arrow{d}[swap]{i}&C\arrow{d}{s}\\
    B\arrow{r}[swap]{k}&D
   \end{tikzcd}
  \]
  commutes. Define $h\colon B\to C$ by for any $b\in \mathrm{im}(i)$ defining $h(b)=j\circ i^{-1} (b)$ (which is well defined because $i$ is injective), and for any $b\notin \mathrm{im}(i)$ defining $h(b)$ to be any element of the preimage of $k(b)$ under $s$ (which is nonempty because $s$ is surjective). Then $h\circ i = j$ and $s \circ h = k$, as required.

  Thus $\mathrm{Inj}\subseteq \llift{\mathrm{Surj}}$ and $\mathrm{Surj}\subseteq \rlift{\mathrm{Inj}}$. Finally we must prove that there are no noninjections in $\llift{\mathrm{Surj}}$ nor nonsurjections in $\rlift{\mathrm{Inj}}$.
  Suppose $i' \colon A\to B$ is not an injection, so there is some $c$ with $c=i'(a)$ and $c=i'(b)$. Let $i'=s\circ i$ be an $(\mathrm{Inj}, \mathrm{Surj})$-factorisation. Then there is no $h$ making 
  \[
   \begin{tikzcd}
    A\arrow{r}{i}\arrow{d}[swap]{i'}&\bullet\arrow{d}{s}\\
    B\arrow{r}[swap]{\id[B]}\arrow{ur}{h}&B
   \end{tikzcd}
  \]
  commute, because $i(a)\neq i(b)$ and $h$ must map $c$ to both $i(a)$ and $i(b)$.

  Likewise, if $s' \colon A\to B$ is not a surjection then pick some $d$ not in its image and let $s'=s\circ i$ be an $(\mathrm{Inj}, \mathrm{Surj})$-factorisation. Then there is no $h$ making
  \[
   \begin{tikzcd}
    A\arrow{r}{\id[A]}\arrow{d}[swap]{i}&A\arrow{d}{s'}\\
    \bullet\arrow{r}[swap]{s}\arrow{ur}{h}&B
   \end{tikzcd}
  \]
  commute, because $h$ must map the nonempty set $s^{-1}(d)$ to the empty set $s'^{-1}(d)$.
\end{example}

Next, we adapt these notions to the setting of symmetric monoidal categories. The problem is that if \cat{C} is a symmetric monoidal category then the subcategories $\llift{\mathcal A}$ and $\rlift{\mathcal A}$ need not be monoidal subcategories. We adapt our definitions accordingly:
\begin{align*}
  f\oslash g
  & \iff \forall A,B\in\mathrm{Ob}(\cat C)\colon f\otimes\mathrm{id}_A\lift g\otimes\mathrm{id}_B \\
  \Rlift{\mathcal{A}} &\; = \{g\in\mathrm{Mor}(\cat C) \mid \forall f\in\mathcal A \colon f\oslash g\} \\
  \Llift{\mathcal{A}} &\; = \{f\in\mathrm{Mor}(\cat C) \mid \forall g\in\mathcal A \colon f\oslash g\}
\end{align*}
Now $\Llift{\mathcal A}$ and $\Rlift{\mathcal A}$ are replete symmetric monoidal subcategories\footnote{The definition of $\oslash$ in terms of $\lift$ resembles complete positivity of quantum processes: for $f$ to be completely positive it is not enough to preserve positivity of states, also $f\otimes \mathrm{id}_A$ must preserve positivity of states.}, and $\Rlift{(-)}$ and $\Llift{(-)}$ form a Galois connection.

\begin{definition}
  A \emph{symmetric monoidal weak factorisation system} is a factorisation system with $\mathcal L=\,\Llift{\mathcal R}$ and $\mathcal R=\Rlift{\mathcal L}$.
\end{definition}

\begin{example}\label{ex:injsurjmonoidal}
  If we consider \cat{Set} to be a symmetric monoidal category under Cartesian products $\times$, then our example $(\mathrm{Inj}, \mathrm{Surj})$ is indeed a symmetric monoidal weak factorisation system. Since $\oslash$ is a stricter property than $\lift$ we have that $\Llift{\mathrm{Surj}} \subseteq \llift{\mathrm{Surj}} = \mathrm{Inj}$ and $\Rlift{\mathrm{Inj}} \subseteq \rlift{\mathrm{Inj}} = \mathrm{Surj}$ so it suffices to check that $i\oslash s$ for all injections $i$ and surjections $s$.
  But this is immediate. If $i$ is an injection then $i\times\id[A]$ is an injection for any set $A$. Likewise, if $s$ is a surjection then $s\times\id[B]$ is an injection for any set $B$. So $i\times\id[A] \mathrel{\lift} s\times\id[B]$ for all $A$ and $B$ and hence $i\oslash s$.
\end{example}

\section{Purity}\label{sec:purity}

Suppose a symmetric monoidal category \cat{C} comes equipped with a morphism $I\rightarrow A$ for each object $A$. We draw these morphisms as $\tinymix$, and call them \emph{completely mixed states}.
Call the collection of all the completely mixed states $\mathcal M'$.
 
 The following definition is our operational description of purity.
 
\begin{definition} 
  Given a symmetric monoidal category equipped with a family of completely mixed states, define the \emph{pure morphisms} to be those in the class $\mathcal P=\Rlift{\mathcal M'}$, and the \emph{mixing morphisms} to be those in the class $\mathcal M=\,\Llift{\mathcal P}$.
\end{definition}
  
The interpretation of the collection $\mathcal M$ is not immediately clear, since the notion of mixing is already captured by the completely mixed states $\mathcal M'$. It will become clear from the examples below that the elements of $\mathcal M$ represent processes which introduce nondeterminism while having essentially no other effect. For example, if $p$ is an isomorphism, then
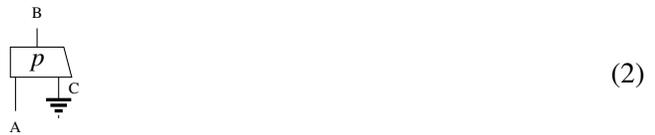
\begin{equation}\label{eq:simplemixing}
    \begin{pic}[xscale=.75,yscale=.5]
      \node[morphism] (a) at (0,.5) {$\;p\;$};
      \draw (a.north) to +(0,.5)node[above]{B};
      \draw (a.south west) to +(0,-.9)node[below]{A};
      \draw (a.south east) to +(0,-.3)node[right]{C} to +(0,-.6)node[ground,scale=-.75]{};
    \end{pic}
\end{equation}
is in $\mathcal M$; we call such maps \emph{simple mixing maps}.
  
Thus we have a pair $(\mathcal M,\mathcal P)$, which might be a factorisation system. If it is a factorisation system, then by construction it will certainly be a symmetric monoidal weak one. Our examples will in fact have the following, even stronger, property.

\begin{definition}
 A category \cat{C} \emph{has purification} if all morphisms are of the form~\eqref{eq:simplemixing}
  for pure $p$.
\end{definition}

For example, Stinespring's theorem says precisely that the category of completely positive maps has purification. If \cat{C} has purification then $(\mathcal M,\mathcal P)$ is a factorisation system because $\mathrm{id}\otimes\tinymix$ is always in $\mathcal M$.

Dualising the above definitions, we can also consider symmetric monoidal categories equipped with families of morphisms $\tinyground$
 which we think of as taking the system $A$ and discarding it. We refer to these morphisms as \emph{discarding effects} and denote the collection of them as $\mathcal D'$. Call the morphisms in the class $\mathcal C= \,\Llift{\mathcal D'}$ \emph{copure}, and define the \emph{discarding morphisms} to be those in the class $\mathcal D=\Rlift{\mathcal C}$. Say \cat{C} has \emph{copurification} if all morphisms are of the form
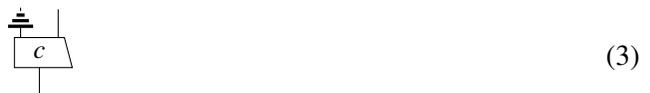
\begin{equation}\label{eq:simplediscarding}
    \begin{pic}[scale=.75]
      \node[morphism] (f) at (0,.5) {$\;c\;$};
      \draw (f.south) to +(0,-.5);
      \draw (f.north east) to +(0,.5);
      \draw (f.north west) to +(0,.2)node[ground,scale=.75]{};
    \end{pic}
\end{equation}
for some copure $c$.

We will see in examples that the copure morphisms represent processes that do not destroy any information, whereas destruction of information is the only effect of the discarding morphisms.
  
Let's decode our abstract definitions of pure and mixing from the language of weak factorisation systems. In diagrammatic notation, $p$ is pure if and only if:
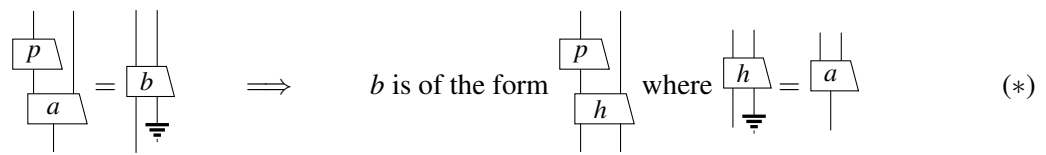
\[
  \tag{$*$}
    \begin{pic}[scale=.75]
      \node[morphism] (1) at (0,-0) {$\;a\;$};
      \node[morphism] (2) at ([yshift=20]1.north west) {$p$};
      \draw (1.north west) to (2.south);
      \draw (1.south) to +(0,-.5);
      \draw (1.north east) to +(0,1.47);
      \draw (2.north) to +(0,.5);
    \end{pic}
    =
    \begin{pic}[scale=.75]
      \node[morphism] (1) at (0,0) {$b$};
      \draw (1.south east) to +(0,-0.5)node[ground,scale=-0.7]{};
      \draw (1.south west) to +(0,-1);
      \draw (1.north east) to +(0,1);
      \draw (1.north west) to +(0,1);
    \end{pic}
    \qquad \implies \qquad
    \text{ $b$ is of the form }
    \begin{pic}[scale=.75]
      \node[morphism] (1) at (0,-0) {$\;h\;$};
      \node[morphism] (2) at ([yshift=20]1.north west) {$p$};
      \draw (1.north west) to (2.south);
      \draw (1.south east) to +(0,-.5);
      \draw (1.south west) to +(0,-.5);
      \draw (1.north east) to +(0,1.47);
      \draw (2.north) to +(0,.5);
    \end{pic}
    \text{ where }
    \begin{pic}[scale=.75]
      \node[morphism] (1) at (0,0) {$h$};
      \draw (1.south east) to +(0,-0.5)node[ground,scale=-0.7]{};
      \draw (1.south west) to +(0,-0.7);
      \draw (1.north east) to +(0,.5);
      \draw (1.north west) to +(0,.5);
    \end{pic}
    =
    \begin{pic}[scale=.75]
      \node[morphism] (1) at (0,0) {$a$};
      \draw (1.north west) to +(0,.5);
      \draw (1.south) to +(0,-.7);
      \draw (1.north east) to +(0,.5);
    \end{pic}
\]
and $m$ is mixing iff
\[
  \tag{$\dagger$}
    \begin{pic}[scale=.75]
      \node[morphism] (1) at (0,-0) {$\;a\;$};
      \node[morphism] (2) at ([yshift=-20]1.south west) {$m$};
      \draw (1.south west) to (2.north);
      \draw (1.north east) to +(0,.5);
      \draw (1.north west) to +(0,.5);
      \draw (1.south east) to +(0,-1.47);
      \draw (2.south) to +(0,-.5);
    \end{pic}
    =
     \begin{pic}[scale=.75]
      \node[morphism] (1) at (0,-0) {$\;b\;$};
      \node[morphism] (2) at ([yshift=20]1.north west) {$p$};
      \draw (1.north west) to (2.south);
      \draw (1.south east) to +(0,-.5);
      \draw (1.south west) to +(0,-.5);
      \draw (1.north east) to +(0,1.47);
      \draw (2.north) to +(0,.5);
    \end{pic}
    \qquad \implies \qquad
    \text{$b$ is of the form }
    \begin{pic}[scale=.75]
      \node[morphism] (1) at (0,0) {$\;h\;$};
      \node[morphism] (2) at ([yshift=-20]1.south west) {$m$};
      \draw (1.south west) to (2.north);
      \draw (1.north east) to +(0,.5);
      \draw (1.north west) to +(0,.5);
      \draw (1.south east) to +(0,-1.47);
      \draw (2.south) to +(0,-.5);
    \end{pic}
    \text{ where }
    \begin{pic}[scale=.75]
      \node[morphism] (1) at (0,-0) {$\;h\;$};
      \node[morphism] (2) at ([yshift=20]1.north west) {$p$};
      \draw (1.north west) to (2.south);
      \draw (1.south east) to +(0,-.5);
      \draw (1.south west) to +(0,-.5);
      \draw (1.north east) to +(0,1.47);
      \draw (2.north) to +(0,.5);
    \end{pic}
    =
    \begin{pic}[scale=.75]
      \node[morphism] (1) at (0,0) {$a$};
      \draw (1.south west) to +(0,-1);
      \draw (1.north east) to +(0,1);
      \draw (1.north west) to +(0,1);
      \draw (1.south east) to +(0,-1);
    \end{pic}
\]
for any pure morphism $p$.  

These two criteria are quite involved, and it is often useful to think of the following two special cases where $a$ and $b$ have particularly simple forms.
If $p$ is pure, then
\[
  \tag{$**$}
    \begin{pic}[scale=.75]
      \node[morphism] (2) at ([yshift=20]1.north west) {$p$};
      \draw (2.south) to +(0,-0.7);
      \draw (2.north) to +(0,.7);
    \end{pic}
    =
    \begin{pic}[scale=.75]
      \node[morphism] (1) at (0,0) {$b$};
      \draw (1.south east) to +(0,-0.3)node[ground,scale=-0.7]{};
      \draw (1.south west) to +(0,-.7);
      \draw (1.north) to +(0,.7);
    \end{pic}
    \qquad \implies \qquad
    \text{ $b$ is of the form }
    \begin{pic}[scale=.75]
      \node[morphism] (1) at (0,-0) {$h$};
      \node[morphism] (2) at ([yshift=20]1.north) {$p$};
      \draw (1.north) to (2.south);
      \draw (1.south east) to +(0,-.5);
      \draw (1.south west) to +(0,-.5);
      \draw (2.north) to +(0,.5);
    \end{pic}
    \text{ where }
    \begin{pic}[scale=.75]
      \node[morphism] (1) at (0,0) {$h$};
      \draw (1.south east) to +(0,-0.3)node[ground,scale=-0.7]{};
      \draw (1.south west) to +(0,-0.7);
      \draw (1.north) to +(0,.7);
    \end{pic}
    =
    \begin{pic}[scale=.75]
      \draw (0,0) to +(0,2);
    \end{pic}
\]
and if $m$ is mixing, then
\[
  \tag{$\dagger\dagger$}
    \begin{pic}[scale=.75]
      \node[morphism] (2) at ([yshift=20]1.north west) {$m$};
      \draw (2.south) to +(0,-0.7);
      \draw (2.north) to +(0,.7);
    \end{pic}
    =
    \begin{pic}[scale=.75]
      \node[morphism] (1) at (0,0) {$p$};
      \draw (1.south east) to +(0,-0.3)node[ground,scale=-0.7]{};
      \draw (1.south west) to +(0,-.7);
      \draw (1.north) to +(0,.7);
    \end{pic}\text{ with $p$ pure }
    \qquad \implies \qquad
    \text{ there exists $h$ with }
    \begin{pic}[scale=.75]
      \node[morphism] (1) at (0,0) {$h$};
      \node[morphism] (2) at ([yshift=-20]1.south) {$m$};
      \draw (1.south) to (2.north);
      \draw (1.north east) to +(0,.5);
      \draw (1.north west) to +(0,.5);
      \draw (2.south) to +(0,-.5);
    \end{pic}
    =
    \begin{pic}[scale=.75]
      \draw (0,2.4) to +(0,-1)node[ground,scale=-0.7]{};
      \draw (-0.3,2.4) to +(0,-2.5);
    \end{pic}
    \,\,\text{ and }\,\,
    \begin{pic}[scale=.75]
      \node[morphism] (1) at (0,0) {$p$};
      \node[morphism] (2) at ([yshift=-20]1.south) {$h$};
      \draw (1.south east) to ([xshift=1]2.north east);
      \draw (1.south west) to ([xshift=-1]2.north west);
      \draw (1.north) to +(0,.5);
      \draw (2.south) to +(0,-.5);
    \end{pic}
    =
    \begin{pic}[scale=.75]
      \draw (0.3,2.4) to +(0,-2.5);
    \end{pic}\,
    \text.
\]
 These special cases are necessary but not sufficient, so if $p$ is pure it certainly obeys $(**)$ and if $m$ is mixing it certainly obeys $(\dagger\dagger)$, but the converse is not necessarily true. However, in many of the particular categories \cat C that we are interested in, we do in fact find that $(**)$ and $(\dagger\dagger)$ are as strong as $(*)$ and $(\dagger)$. So when it comes to calculate the classes $\mathcal P$ and $\mathcal M$ for a particular category \cat{C}, one can often proceed by finding the classes of morphisms which obey $(**)$ and $(\dagger\dagger)$, and then verifying that these morphisms further satisfy $(*)$ and $(\dagger)$.
  
 \section{Examples}\label{sec:examples}

The main claim of this paper is that our definitions of `pure', `mixing', `copure' and `discarding' are in accordance with expectations when interpreting morphisms as processes: if a symmetric monoidal category \cat C equipped with a family of morphisms $\tinymix$ is interpreted as a collection of processes in which $\tinymix$ represents a completely mixed state, then the family of morphisms which are deterministic in this interpretation should be those in the class $\mathcal P$.
We will give four examples.

\subsection{Probabilistic}

Let \cat{FStoch} be the symmetric monoidal category whose objects are finite sets (with the Cartesian product as $\otimes$) and whose morphisms are matrices of nonnegative reals with the domain labelling the columns and the codomain labelling the rows (with matrix multiplication as $\circ$, and Kronecker product of matrices as $\otimes$). Interpret a matrix with entries $(f_{ij})$ as a process which, if it begins in state $i$, has probability $f_{ij}$ of going to state $j$.\footnote{
We will forego, in all examples, the usual normalisation demand that the columns sum to $1$. This is to avoid having to deal with the completely mixed state on the empty set, which cannot be normalised. Thus in this example the completely mixed state on a set with $n$ elements will be a column vector with entries $1$ rather than $1/n$.}
The morphisms $\,\tinymix$ and $\,\tinyground$ are given by the column vector of $1$s and the row vector of $1$s.

\begin{lemma}\label{lem:probabilistic:system}
  Suppose that $\mathcal R$ is the collection matrices which have at most one nonzero entry in each column, and $\mathcal L$ is the collection of matrices which have exactly one nonzero entry in each row and at least one nonzero entry in each column. Then $(\mathcal L,\mathcal R)$ is a factorisation system, and $m\oslash p$ for $m\in\mathcal L$, $p\in\mathcal R$.
\end{lemma}
\begin{proof}
  See Appendix~\ref{sec:proofs}.
\end{proof}

\begin{lemma}\label{lem:probabilistic:pure}
  The symmetric monoidal subcategory $\mathcal P$ consists of precisely those matrices which have at most one nonzero entry in each column. Dually, $\mathcal C$ consists of precisely those matrices which have at most one nonzero entry in each row.
\end{lemma}
\begin{proof}
Suppose $p:X\to Y$ is pure, and suppose that $p_{y_0x_0}\neq 0$ and $p_{y_1x_0}\neq 0$. In $(**)$, let $b:X\times\{0,1\}\to Y$ be defined by setting $b_{yx0}=p_{yx}$ for all $x,y$ except that $b_{y_1x_00}=0$, and $b_{yx1}=0$ for all $x,y$ except that $b_{y_1x_01}=p_{y_1x_0}$. Then $h_{x_0x_01}\neq 0$, so $b_{y_0x_01}=(p\circ h)_{y_0x_01}\neq 0$. Hence $y_0=y_1$. So $p$ has at most one nonzero entry in each column.

Conversely, suppose that $p$ has at most one nonzero entry in each column. Then by Lemma~\ref{lem:probabilistic:system} $p\oslash(\tinymix\otimes \mathrm{id}_B)$, so $p$ is pure.

 The second statement follows by taking the transpose of all matrices involved.
\end{proof}

\begin{lemma}\label{lem:probabilistic:discarding}
  The symmetric monoidal subcategory $\mathcal M$ consists of precisely those matrices which have exactly one nonzero entry in each row and at least one nonzero entry in each column. Dually, $\mathcal D$ consists of precisely those matrices which have exactly one nonzero entry in each column and at least one nonzero entry in each row.
\end{lemma}
\begin{proof}
 Suppose $m:X\to Y$ is in $\mathcal M$. In $(\dagger \dagger)$, define $p:X\times Y\to Y$ by $p_{yxy}=m_{yx}$ and $p_{y'xy}=0$ if $y\neq y'$. Let $y\in Y$. Then $(p\circ h)_{yy}=1$, so there are $y'\in Y$ and $x\in X$ with $h_{xy'y}\neq 0$. Now $(h\circ m)_{xy'x}=(\mathrm{id}_X\times\tinymix_Y)_{xy'x}=1$, and so $m_{yx}\neq 0$. If also $m_{yx'}\neq 0$, then $(h\circ m)_{x'y'x}\neq 0$ and so $x'=x$. Thus $m$ has exactly one nonzero entry in each row.
 
 If $Y\neq \emptyset$ then $\mathrm{id}_X\times\tinymix_Y$ has at least one nonzero entry in each column. Since $h\circ m=\mathrm{id}_X\times\tinymix_Y$ the map $m$ must also have at least one nonzero entry in each column. If $Y=\emptyset$ then the same proof applies after adding an element to the ancillary system.
  
 Conversely, suppose that a relation $m$ has exactly one nonzero entry in each row. Then by Lemma~\ref{lem:probabilistic:system} and Lemma~\ref{lem:probabilistic:pure} $p\oslash m$ for any pure $p$, so $m\in \mathcal M$.

 The second statement follows by taking the transpose of all matrices involved.
\end{proof}

\begin{theorem}\label{thm:probabilistic:purification}
  The category \cat{FStoch} has purification and copurification. Hence the pairs $(\mathcal M, \mathcal P)$ and $(\mathcal C, \mathcal D)$ are symmetric monoidal weak factorisation systems.
  \end{theorem}
\begin{proof}
  Any matrix $f:I\to J$ of nonnegative reals factors as follows. Let the ancilla system be the set $J$. Let $p$ be the matrix with $((i,j),j')$-entry $\delta_{jj'}f_{ij}$; this is pure by Lemma~\ref{lem:probabilistic:pure}. Then $f$ is given by~\eqref{eq:matrix}.
  The dual statement is obtained by transposing $f$.
\end{proof}

\subsection{Possibilistic}

Let $\cat{FRel}$ be the symmetric monoidal category whose objects are finite sets (with the Cartesian product as $\otimes$) and whose morphisms are relations between the source and target. Interpret a relation $r \colon A\rightarrow B$ as a nondeterministic process that might send an $a\in A$ to any of the $b\in B$ to which it is related.\footnote{In this case the absence of normalisation means that there might be $a\in A$ not related to anything.} We take $\tinymix$ to be the relation that relates the point to everything in the target, and $\tinyground$ to be the relation that relates everything in the source to the point.
\begin{lemma}
\label{lem:possibilistic:system}
  Suppose that $\mathcal R$ is the collection of partial functions (\textit{i.e.}\ relations in which everything in the source is related to at most one element of the target) and $\mathcal L$ is the collection of surjective, injective and total relations (\textit{i.e.}\ those in which everything in the target is related to exactly one thing in the source and everything in the source is related to at least one thing in the target). Then $(\mathcal L,\mathcal R)$ is a factorisation system, and $m\oslash p$ for $m\in\mathcal L$, $p\in\mathcal R$.
\end{lemma}
\begin{proof}
  See Appendix~\ref{sec:proofs}.
\end{proof}

\begin{lemma}
\label{lem:possibilistic:pure}
  The symmetric monoidal subcategory $\mathcal P$ consists precisely of the partial functions. Dually, $\mathcal C$ consists precisely of the injective relations (\textit{i.e.}\ those in which everything in the target is related to at most one element of the source).
\end{lemma}
\begin{proof}
  Suppose $p \in \mathcal{P}$, and $(x_0,y_0),(x_0,y_1) \in p$.
  In $(**)$, set $b=\{(0,x,y) \mid (x,y) \in p, (x,y)\neq(x_0,y_1)\} \cup \{(1,x_0,y_1)\}$. Then $(1,x_0,x_0) \in h$, so $(1,x_0,y_0) \in p \circ h = b$. Hence $y_0=y_1$. Thus $p$ is a partial function.

  Conversely, suppose that $p$ is a partial function. Then by Lemma~\ref{lem:possibilistic:system} $p\oslash(\tinymix\otimes \mathrm{id}_B)$, so $p$ is pure.

  The second statement follows by taking the dual of all relations involved.
\end{proof}

\begin{lemma}
  The symmetric monoidal subcategory $\mathcal M$ consists of precisely the surjective, injective and total relations. Dually, $\mathcal D$ is precisely the surjective total functions (\textit{i.e.}\ the relations in which everything in the source is related to exactly one thing in the target and everything in the target is related to at least one thing in the source).
\end{lemma}
\begin{proof}
  Suppose $m \colon X \to Y$ is in $\mathcal{M}$. In $(\dag\dag)$, set $p=\{(x,y,y) \mid (x,y) \in m\}$. Let $y \in Y$. Then $(y,y) \in p \circ h$, so there are $y'\in Y$ and $x \in X$ with $(y,x,y') \in h$.
  Now $(x,x,y') \in   \id[X] \times \tinymix_Y= h \circ m$, and so $(x,y) \in m$.  If also $(x',y) \in m$, then $(x',x,y') \in h \circ m =   \id[X] \times \tinymix_Y$ and so $x'=x$. Thus $m$ is surjective and injective.

 If $Y\neq \emptyset$ then $\mathrm{id}_X\times\tinymix_Y$ is total. Since $h\circ m=\mathrm{id}_X\times\tinymix_Y$ the map $m$ must also be total. If $Y=\emptyset$ then the same proof applies after adding an element to the ancillary system.

  Conversely, suppose that a relation $m$ is surjective and injective. Then by Lemma~\ref{lem:possibilistic:system} and Lemma~\ref{lem:possibilistic:pure} $p\oslash m$ for any pure $p$, so $m\in\mathcal M$.
  
  The second statement follows by taking dual relations.
\end{proof}

\begin{theorem}
  The category \cat{FRel} has purification and copurification. Hence the pairs $(\mathcal M, \mathcal P)$ and $(\mathcal C, \mathcal D)$ are symmetric monoidal weak factorisation systems.
  \end{theorem}
\begin{proof}
  Any relation $r \colon A \to B$ factors as the introduction of a completely mixed state on $B$, given by the relation $A \to A \times B$ that relates $a \in A$ to $(a,b)$ for any $b \in B$, followed by a partial function $A \times B \to B$ that sends $(a,b)$ to $b$ when $(a,b) \in r$. 
  Dually, $r$ factors as an injective relation $A \to A \times B$ that relates $a \in A$ to $(a,b)$ if $(a,b) \in r$, followed by discarding $A$, \textit{i.e.}\ the relation $A \times B \to B$ that relates $(a,b)$ to $b$.
\end{proof}

\subsection{Deterministic}

Let $\cat{FSet}$ be the category whose objects are finite sets (with the Cartesian product as $\otimes$) and whose morphisms are functions. Interpret a function $f\colon A\rightarrow B$ as a deterministic process that sends $a\in A$ to $f(b)\in B$.\footnote{As we do not consider mixing, the category in this example is normalised. The unnormalised version has \emph{partial} functions.} Take $\tinyground$ to be the function that maps everything in its source to the point (in fact we had to chose this function, since $I$ is terminal). There is no notion of a completely mixed state in this category.

\begin{lemma}
  The symmetric monoidal subcategory $\mathcal C$ consists precisely of the injections.
\end{lemma}
\begin{proof}
  Unfolding the definitions, $f \colon A \to B$ is in $\mathcal{C}$ when for all finite sets $C,D,E$ and functions $g,h$ making the following square commute there is a fill-in $k$:
  \[
   \begin{tikzcd}
    A\times C\arrow{r}{g}\arrow{d}[swap]{f \times \id[C]}&D \times E\arrow{d}{\pi_2}\\
    B \times C\arrow{r}[swap]{h}\arrow[dashed]{ur}{k}&E
   \end{tikzcd}
  \]
  If this is the case, taking $D=A$, $E=C=1$, $g=\id$, and $h=\pi_2$ shows that $f$ has a left-inverse and so is injective.
  For the converse we may assume that $D$ is nonempty and fix $d \in D$. When $f$ is injective, we can define $k(b,c)=g(a,c)$ if $f(a)=b$, and $k(b,c)=(d,h(b,c))$ otherwise.
\end{proof}

\begin{lemma}
  The symmetric monoidal subcategory $\mathcal D$ consists precisely of the surjections.
\end{lemma}
\begin{proof}
  Unfolding definitions with the previous lemma, $g \colon C \to D$ is in $\mathcal{D}$ when for all finite sets $A$, $B$, $E$, $F$, injections $f \colon A \rightarrowtail B$, and functions $h$, $k$ making the following square commute there is a fill-in:
  \[
   \begin{tikzcd}
    A\times E\arrow{r}{h}\arrow[>->]{d}[swap]{f \times \id[E]}&C \times F\arrow{d}{g \times \id[F]}\\
    B \times E\arrow{r}[swap]{k}\arrow[dashed]{ur}{l}&D \times F
   \end{tikzcd}
  \]
  If this is the case, taking $A=\emptyset$, $B=D$, and $E=F=1$ shows that $g$ has a right-inverse and so is surjective.
  Conversely, if $g$ is surjective, then $g \in \rlift{\mathcal{C}}$ by Example~\ref{ex:injsurjweakfactorisationsystem}, so certainly $g \in \mathcal{D}=\Rlift{\mathcal{C}}$.
\end{proof}

The previous two lemmas prove that $(\mathcal C, \mathcal D)$ is just $(\mathrm{Inj},\mathrm{Surj})$, our canonical example of a weak factorisation system. Example~\ref{ex:injsurjmonoidal} showed that it is also a symmetric monoidal weak factorisation system. 
  
\begin{theorem}
  The category \cat{FSet} has copurification.
\end{theorem}
\begin{proof}
  Any function $f \colon A \to B$ factors as the injection $A \rightarrowtail A \times B$ given by $a \mapsto (a,f(a))$, followed by the discarding map $\pi_2 \colon A \times B \to B$.
\end{proof}

\subsection{Quantum}

Let \cat{Quant} be the category of finite-dimensional Hilbert spaces and completely positive maps\footnote{Since we are not restricting our attention to normalised maps we do not demand that our maps preserve trace.}. Let $\tinymix$ be the state with density matrix given by the identity, and let $\tinyground$ take the trace.
There is a canonical functor $F \colon \cat{FHilb}\rightarrow\cat{Quant}$, and the pure morphisms in \cat{Quant} are usually defined to be the ones in the image of this functor~\cite{selinger:cpm,coeckeheunen:cp}. We will show that our definitions agree with this.


\begin{lemma}\label{lem:minimalstinespring}
  \begin{enumerate}[label={(\alph*)}]
  \item \label{lem:minimal:exists} Every completely positive map has a \emph{minimal} Stinespring dilation $p$ as in~\eqref{eq:stinespring}, 
   such that any other Stinespring dilation $p'$ allows a coisometry\footnote{A linear map between two Hilbert spaces is a \emph{coisometry} iff $j\circ j^\dagger=\mathrm{id}$. Its image in \cat{Quant} satisfies $j \circ \tinymix = \tinymix$.}  $j \in \cat{FHilb}$ satisfying:
    \begin{equation}\label{eq:minimalstinespring}
     \begin{pic}[xscale=.75,yscale=.5]
       \node[morphism] (f) at (0,-.5) {$\;p\;$};
       \node[morphism] (i) at ([yshift=-21]f.south east) {$\;j\;$};
       \draw (f.north) to +(0,.5);
       \draw (f.south west) to +(0,-1.5);
       \draw (f.south east) to (i.north);
       \draw (i.south) to +(0,-.5);
     \end{pic}\quad
     = \quad
     \begin{pic}[xscale=.75,yscale=.5]
       \node[morphism] (f) at (0,-.5) {$\;p'\;$};
       \draw (f.north) to +(0,.5);
       \draw (f.south west) to +(0,-1.5);
       \draw (f.south east) to +(0,-1.5);
     \end{pic}
 \end{equation}
\item \label{lem:minimal:epi} A Stinespring dilation $p$ as in~\eqref{eq:stinespring} is minimal if and only if the following morphism is a surjection:
   
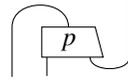
\[
     \begin{pic}[xscale=.75,yscale=.5]
       \node[morphism] (f) at (0,-.5) {$\;p\;$};
       \draw (f.north) to[in=90,out=90,looseness=2] +(-1,0) to +(0,-1.325);
       \draw (f.south west) to +(0,-.5);
       \draw (f.south east) to[in=-90,out=-90,looseness=1.5] +(0.7,0) to +(0,1.5);
     \end{pic}
 \]
\item \label{lem:extenddilation} 
  If $p$ and $p'$ are Stinespring dilations as in~\eqref{eq:stinespring} of the same map, and their ancillas $C$ and $C'$ satisfy $\dim(C) \leq \dim(C')$, 
 then there is a coisometry $j$ satisfying~\eqref{eq:minimalstinespring}.
\end{enumerate}
\end{lemma}
\begin{proof}
  See Appendix~\ref{sec:proofs}.
\end{proof}

\begin{lemma}\label{lem:quantum:pure}
  The symmetric monoidal subcategory $\mathcal P$ consists of precisely those morphisms of the form $F(f)$ for $f\in\mathrm{Mor}(\cat{FHilb})$.
\end{lemma}
\begin{proof}
  First we assume $p$ pure and show it is in the image of $F$. Let $P$ be a Stinespring dilation of $p$.
  Without loss of generality, assume $p$ nonzero, so all ancillas are nonzero.
  As $p$ is pure, $(**)$ gives $h$ with:
  
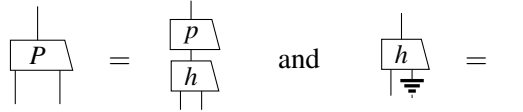
\[
  \begin{pic}[xscale=.75,yscale=.5]
       \node[morphism] (f) at (0,-.5) {$\;P\;$};
       \draw (f.north) to +(0,.9);
       \draw (f.south west) to +(0,-.9);
       \draw (f.south east) to +(0,-.9);
     \end{pic}\quad
     =\quad
    \begin{pic}[xscale=.75,yscale=.5]
      \node[morphism] (1) at (0,-0) {$h$};
      \node[morphism] (2) at ([yshift=20]1.north) {$p$};
      \draw (1.north) to (2.south);
      \draw (1.south east) to +(0,-.5);
      \draw (1.south west) to +(0,-.5);
      \draw (2.north) to +(0,.5);
    \end{pic}\qquad
    \text{and}\qquad
    \begin{pic}[xscale=.75,yscale=.5]
      \node[morphism] (1) at (0,0) {$h$};
      \draw (1.south east) to +(0,-0.3)node[ground,scale=-0.7]{};
      \draw (1.south west) to +(0,-0.7);
      \draw (1.north) to +(0,.7);
    \end{pic}\quad
    =\quad
    \begin{pic}[xscale=.75,yscale=.5]
      \draw (0,0) to +(0,2);
    \end{pic}
    \]
    Let $H$ be a Stinespring dilation of $h$. By Lemma~\ref{lem:minimalstinespring}\ref{lem:minimal:epi} the identity is clearly a minimal one. So
    
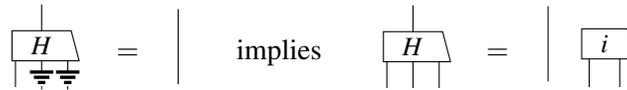
\[
    \begin{pic}[xscale=.75,yscale=.5]
      \node[morphism] (1) at (0,0) {$\;H\;$};
      \draw (1.south east) to +(0,-0.3)node[ground,scale=-0.7]{};
      \draw (1.south) to +(0,-0.3)node[ground,scale=-0.7]{};
      \draw (1.south west) to +(0,-0.7);
      \draw (1.north) to +(0,.7);
    \end{pic}\quad
    =\quad
    \begin{pic}[xscale=.75,yscale=.5]
      \draw (0,0) to +(0,2);
    \end{pic}\qquad
    \text{implies}\qquad
    \begin{pic}[xscale=.75,yscale=.5]
      \node[morphism] (1) at (0,0) {$\;H\;$};
      \draw (1.south east) to +(0,-0.7);
      \draw (1.south) to +(0,-0.7);
      \draw (1.south west) to +(0,-0.7);
      \draw (1.north) to +(0,.7);
    \end{pic}\quad
    =\quad
    \begin{pic}[xscale=.75,yscale=.5]
      \draw (0,-1) to +(0,2);
      \node[morphism] (1) at (1,0) {$\;i\;$};
      \draw (1.south east) to +(0,-0.75);
      \draw (1.south west) to +(0,-0.75);
    \end{pic}
    \]
    for some coisometry $i$. Therefore:
    
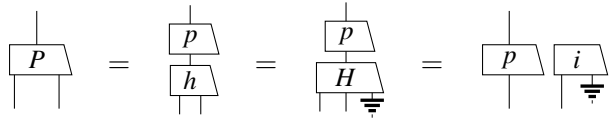
\[
    \begin{pic}[xscale=.75,yscale=.5]
       \node[morphism] (f) at (0,-.5) {$\;P\;$};
       \draw (f.north) to +(0,.9);
       \draw (f.south west) to +(0,-.9);
       \draw (f.south east) to +(0,-.9);
     \end{pic}\quad
     =\quad
    \begin{pic}[xscale=.75,yscale=.5]
      \node[morphism] (1) at (0,-0) {$h$};
      \node[morphism] (2) at ([yshift=20]1.north) {$p$};
      \draw (1.north) to (2.south);
      \draw (1.south east) to +(0,-.5);
      \draw (1.south west) to +(0,-.5);
      \draw (2.north) to +(0,.5);
    \end{pic}\quad
     =\quad\begin{pic}[xscale=.75,yscale=.5]
      \node[morphism] (1) at (0,-0) {$\;H\;$};
      \node[morphism] (2) at ([yshift=20]1.north) {$p$};
      \draw (1.north) to (2.south);
      \draw (1.south east) to +(0,-0.2)node[ground,scale=-0.7]{};
      \draw (1.south) to +(0,-.5);
      \draw (1.south west) to +(0,-.5);
      \draw (2.north) to +(0,.5);
    \end{pic}\quad
     =\quad\begin{pic}[xscale=.75,yscale=.5]
       \node[morphism] (f) at (0,-.5) {$\;p\;$};
       \draw (f.north) to +(0,.9);
       \draw (f.south) to +(0,-.9);
      \node[morphism] (1) at (1.2,-0.5) {$\;i\;$};
      \draw (1.south east) to +(0,-0.3)node[ground,scale=-0.7]{};
      \draw (1.south west) to +(0,-0.9);
     \end{pic}
    \]
Now pick any state $a \colon \mathbb{C} \to A$ in $\cat{FHilb}$ that makes $i \circ (a \otimes \tinymix)$ nonzero.
  Any scalar in \cat{Quant} lies in the image of $F$, as do $P$ and $a$. Hence so does:
  \[\begin{pic}[scale=.75]
       \node[morphism] (f) at (0,-.5) {$\;P\;$};
       \node[morphism] (a) at ([yshift=-10]f.south east) {$a$};
       \draw (f.north) to +(0,.9);
       \draw (f.south west) to +(0,-.9);
       \draw (f.south east) to (a.north);
     \end{pic}\quad
     \left(\begin{pic}[scale=.75]
      \node[morphism] (1) at (1.2,-0.5) {$\;i\;$};
       \node[morphism] (a) at ([yshift=-10]1.south west) {$a$};
      \draw ([xshift=5]1.south east) to +(0,-0.3)node[ground,scale=-0.7]{};
      \draw (1.south west) to (a.north);
     \end{pic}\right)^{-1}
     \quad=\quad
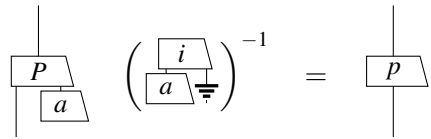
     \begin{pic}[scale=.75]
       \node[morphism] (f) at (0,-.5) {$\;p\;$};
       \draw (f.north) to +(0,.9);
       \draw (f.south) to +(0,-.9);
     \end{pic}
     \]
     
  For the converse, we assume $p\in\cat{FHilb}$ and prove it pure. Suppose
  $(p \otimes \id) \circ a = b \circ (\id \otimes \tinyground)$.
    We want to find some $h$ satisfying $(*)$. Without loss of generality we may assume $b$ nonzero. 
    If $A$ and $B$ are Stinespring dilations of $a$ and $b$, then:
     
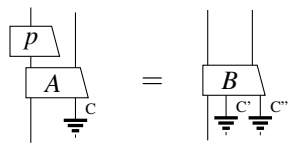
\[
      \begin{pic}[xscale=.75,yscale=.5]
      \node[morphism] (1) at (0,-0) {$\;A\;$};
      \node[morphism] (2) at ([yshift=20]1.north west) {$p$};
      \draw (1.north west) to (2.south);
      \draw (1.south west) to +(0,-1.2);
      \draw (1.south east) to +(0,-0.3)node[right]{C} to +(0,-0.6)node[ground,scale=-0.7]{};
      \draw (1.north east) to +(0,1.47);
      \draw (2.north) to +(0,.5);
    \end{pic}
\quad    =\quad
    \begin{pic}[xscale=.75,yscale=.5]
      \node[morphism] (1) at (0,0) {$\;B\;$};
      \draw ([xshift=4]1.south east) to +(0,-0.3)node[right]{C''} to +(0,-0.6)node[ground,scale=-0.7]{};
      \draw ([xshift=-2]1.south) to +(0,-0.3)node[right]{C'} to +(0,-0.6)node[ground,scale=-0.7]{};
      \draw (1.south west) to +(0,-1.2);
      \draw (1.north east) to +(0,1.47);
      \draw (1.north west) to +(0,1.47);
    \end{pic}
    \]
    Pick such $A$ and $B$  with the further property that $\mathrm{dim}(C)\leq\mathrm{dim}(C')\mathrm{dim}(C'')$. This can be done by picking any $A$ and $B$ and then enlarging $C''$ by applying an arbitrary coisometry from a space with dimension greater than $\mathrm{dim}(C)/\mathrm{dim}(C')$. 
    Then Lemma~\ref{lem:minimalstinespring}\ref{lem:extenddilation} provides a coisometry $j$ satisfying
    $(p \otimes \id) \circ A \circ (\id \otimes j) = B$.
    Define $h = A \circ (\id \otimes j) \circ (\id \otimes \id \otimes \tinymix)$.
    Then
    \[
      \begin{pic}[xscale=.75,yscale=.5]
      \node[morphism] (1) at (0,0) {$\;h\;$};
      \node[morphism] (2) at ([yshift=20]1.north west) {$p$};
      \draw (1.north west) to (2.south);
      \draw (1.south east) to +(0,-1.47);
      \draw (1.south west) to +(0,-1.47);
      \draw (1.north east) to +(0,1.47);
      \draw (2.north) to +(0,0.5);
    \end{pic}\quad    =\quad
      \begin{pic}[xscale=.75,yscale=.5]
      \node[morphism] (1) at (0,-0) {$\;A\;$};
      \node[morphism] (3) at ([yshift=-20]1.south east) {$j$};
      \node[morphism] (2) at ([yshift=20]1.north west) {$p$};
      \draw (1.north west) to (2.south);
      \draw (2.north) to +(0,0.5);
      \draw (1.south west) to +(0,-1.47);
      \draw (1.south east) to (3.north);
      \draw (1.north east) to +(0,1.47);
      \draw (3.south east) to +(0,-0.2)node[ground,scale=-0.7]{};
      \draw (3.south west) to +(0,-0.5);
    \end{pic}   \quad    =\quad
    \begin{pic}[xscale=.75,yscale=.5]
      \node[morphism] (1) at (0,0) {$\;B\;$};
      \draw ([xshift=2]1.south east) to +(0,-0.5)node[ground,scale=-0.7]{};
      \draw ([xshift=-2]1.south) to +(0,-1.47);
      \draw (1.south west) to +(0,-1.47);
      \draw (1.north east) to +(0,1.47);
      \draw (1.north west) to +(0,1.47);
    \end{pic}\quad    =\quad
    \begin{pic}[xscale=.75,yscale=.5]
      \node[morphism] (1) at (0,0) {$\;b\;$};
      \draw (1.south east) to +(0,-1.47);
      \draw (1.south west) to +(0,-1.47);
      \draw (1.north east) to +(0,1.47);
      \draw (1.north west) to +(0,1.47);
    \end{pic}
    \]
    and
     \[
      \begin{pic}[xscale=.75,yscale=.5]
      \node[morphism] (1) at (0,0) {$\;h\;$};
      \draw (1.south east) to +(0,-0.5)node[ground,scale=-0.7]{};
      \draw (1.south west) to +(0,-1.47);
      \draw (1.north east) to +(0,0.5);
      \draw (1.north west) to +(0,0.5);
    \end{pic}\quad    =\quad
      \begin{pic}[xscale=.75,yscale=.5]
      \node[morphism] (1) at (0,-0) {$\;A\;$};
      \node[morphism] (3) at ([yshift=-20]1.south east) {$j$};
      \draw (1.north west) to +(0,0.5);
      \draw (1.south west) to +(0,-1.47);
      \draw (1.south east) to (3.north);
      \draw (1.north east) to +(0,0.5);
      \draw ([xshift=2]3.south east) to +(0,-0.2)node[ground,scale=-0.7]{};
      \draw ([xshift=-2]3.south west) to +(0,-0.2)node[ground,scale=-0.7]{};
    \end{pic}   \quad    =\quad
      \begin{pic}[xscale=.75,yscale=.5]
      \node[morphism] (1) at (0,-0) {$\;A\;$};
      \draw (1.north west) to +(0,0.5);
      \draw (1.south west) to +(0,-1.47);
      \draw (1.south east) to +(0,-0.2)node[ground,scale=-0.7]{};
      \draw (1.north east) to +(0,0.5);
    \end{pic}   \quad    =\quad
      \begin{pic}[xscale=.75,yscale=.5]
      \node[morphism] (1) at (0,-0) {$\;a\;$};
      \draw (1.north west) to +(0,0.5);
      \draw (1.south) to +(0,-1.47);
      \draw (1.north east) to +(0,0.5);
    \end{pic}   
    \]
    as required.
\end{proof}

\begin{lemma}
  The symmetric monoidal subcategory $\mathcal M$ consists of only the simple mixing maps.
\end{lemma}
\begin{proof}
  Let $m\in\mathcal M$, and let $M$ be a minimal Stinespring dilation of $m$. 
  We will exhibit an inverse of $M$. By Lemma~\ref{lem:quantum:pure}, $M$ is pure, so  $(\dagger\dagger)$ provides $h$ satisfying:
  
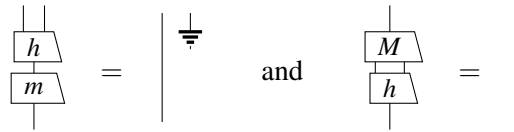
\[
  \begin{pic}[xscale=.75,yscale=.5]
      \node[morphism] (2) at (0,0) {$m$};
      \node[morphism] (1) at ([yshift=20]1.north) {$h$};
      \draw (2.south) to +(0,-0.7);
      \draw (2.north) to (1.south);
      \draw (1.north west) to +(0,0.7);
      \draw (1.north east) to +(0,0.7);
    \end{pic}\quad
    =\quad
  \begin{pic}[xscale=.75,yscale=.5]
    \draw (0,0) to (0,-2.9);
    \draw (0.5,0) to +(0,-0.5)node[ground,scale=-0.7]{};
  \end{pic}\qquad\text{and}\qquad
  \begin{pic}[xscale=.75,yscale=.5]
      \node[morphism] (2) at (0,0) {$h$};
      \node[morphism] (1) at ([yshift=20]2.north) {$M$};
      \draw (2.south) to +(0,-0.7);
      \draw (1.north) to +(0,0.7);
      \draw ([xshift=2]1.south west) to ([xshift=-2]2.north west);
      \draw ([xshift=-2]1.south east) to ([xshift=2]2.north east);
  \end{pic}\quad
  =\quad
  \begin{pic}[xscale=.75,yscale=.5]
    \draw (0,0) to (0,-2.9);
  \end{pic}
  \]
    Let $H$ be a Stinespring dilation of $h$.
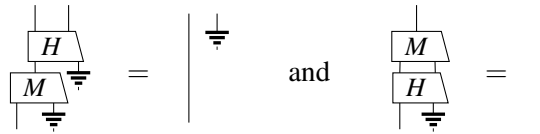
    \[
    \begin{pic}[xscale=.75,yscale=.5]
      \node[morphism] (2) at (0,0) {$M$};
      \node[morphism] (1) at ([yshift=20,xshift=8.4]2.north) {$H$};
      \draw (2.south west) to +(0,-0.7);
      \draw (2.south east) to +(0,-0.3)node[ground,scale=-0.7]{};
      \draw ([xshift=5]1.south east) to +(0,-0.3)node[ground,scale=-0.7]{};
      \draw (2.north) to (1.south west);
      \draw (1.north west) to +(0,0.7);
      \draw (1.north east) to +(0,0.7);
    \end{pic}\quad
    =\quad
    \begin{pic}[xscale=.75,yscale=.5]
      \draw (0,0) to (0,-2.9);
      \draw (0.5,0) to +(0,-0.5)node[ground,scale=-0.7]{};
    \end{pic}\qquad\text{and}\qquad
    \begin{pic}[xscale=.75,yscale=.5]
      \node[morphism] (2) at (0,0) {$H$};
      \node[morphism] (1) at ([yshift=20]2.north) {$M$};
      \draw (2.south west) to +(0,-0.7);
      \draw (2.south east) to +(0,-0.3)node[ground,scale=-0.7]{};
      \draw (1.north) to +(0,0.7);
      \draw ([xshift=1]1.south west) to (2.north west);
      \draw (1.south east) to ([xshift=1]2.north east);
    \end{pic}\quad
    =\quad
    \begin{pic}[xscale=.75,yscale=.5]
      \draw (0,0) to (0,-2.9);
    \end{pic}
    \]
    By Lemma~\ref{lem:minimalstinespring}\ref{lem:minimal:epi} the right hand side of each of these equations is a minimal Stinespring dilation, so there are coisometries $j$ and $j'$ satisfying:
    
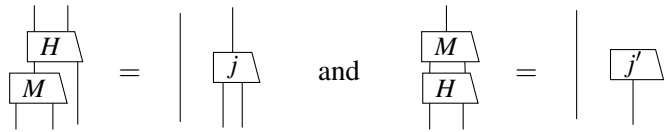
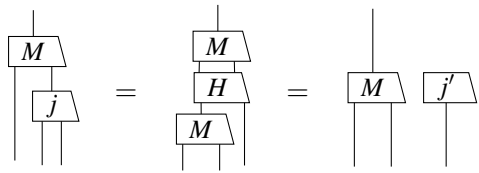
\[
    \begin{pic}[xscale=.75,yscale=.5]
      \node[morphism] (2) at (0,0) {$M$};
      \node[morphism] (1) at ([yshift=20,xshift=8.4]2.north) {$H$};
      \draw (2.south west) to +(0,-0.7);
      \draw (2.south east) to +(0,-0.7);
      \draw ([xshift=5]1.south east) to +(0,-1.67);
      \draw (2.north) to (1.south west);
      \draw (1.north west) to +(0,0.7);
      \draw (1.north east) to +(0,0.7);
    \end{pic}\quad
    =\quad
    \begin{pic}[xscale=.75,yscale=.5]
      \draw (0,0) to (0,-2.9);
    \end{pic}\quad
    \begin{pic}[xscale=.75,yscale=.5]
      \node[morphism] (j) at (1,-1.5) {$j$};
      \draw (j.north) to +(0,1.2);
      \draw (j.south east) to +(0,-1.2);
      \draw (j.south west) to +(0,-1.2);
    \end{pic}\qquad\text{and}\qquad
    \begin{pic}[xscale=.75,yscale=.5]
      \node[morphism] (2) at (0,0) {$H$};
      \node[morphism] (1) at ([yshift=20]2.north) {$M$};
      \draw (2.south west) to +(0,-0.7);
      \draw (2.south east) to +(0,-0.7);
      \draw (1.north) to +(0,0.7);
      \draw ([xshift=1]1.south west) to (2.north west);
      \draw (1.south east) to ([xshift=1]2.north east);
    \end{pic}\quad
    =\quad
    \begin{pic}[xscale=.75,yscale=.5]
      \draw (0,-1.45) to (0,1.45);
      \node[morphism] (j) at (1,0) {$j'$};
      \draw (j.south) to +(0,-1.2);
    \end{pic}
    \]
    Hence
    \[
    \begin{pic}[xscale=.75,yscale=.5]
      \node[morphism] (2) at (0,0) {$j$};
      \node[morphism] (1) at ([yshift=30,xshift=-9.3]2.north) {$M$};
      \draw (2.south west) to +(0,-1.17);
      \draw (2.south east) to +(0,-1.17);
      \draw ([xshift=0]1.south west) to +(0,-2.5);
      \draw (2.north) to (1.south east);
      \draw (1.north) to +(0,0.7);
    \end{pic}\quad=\quad
    \begin{pic}[xscale=.75,yscale=.5]
      \node[morphism] (2) at (0,0) {$M$};
      \node[morphism] (1) at ([yshift=20,xshift=8.4]2.north) {$H$};
      \node[morphism] (3) at ([yshift=20]1.north) {$M$};
      \draw (2.south west) to +(0,-0.7);
      \draw (2.south east) to +(0,-0.7);
      \draw ([xshift=5]1.south east) to +(0,-1.67);
      \draw (2.north) to (1.south west);
      \draw (3.north) to +(0,0.7);
      \draw ([xshift=-1]1.north west) to (3.south west);
      \draw (1.north east) to ([xshift=-1]3.south east);
    \end{pic}\quad=\quad
    \begin{pic}[xscale=.75,yscale=.5]
      \node[morphism] (2) at (0,0) {$M$};
      \draw (2.south west) to +(0,-1.7);
      \draw (2.south east) to +(0,-1.7);
      \draw (2.north) to +(0,1.7);
      \node[morphism] (j) at (1.3,0) {$j'$};
      \draw (j.south) to +(0,-1.7);
    \end{pic}
    \]
  Lemma~\ref{lem:minimalstinespring}\ref{lem:minimal:epi} lets us remove $M$ to see $j = \id \otimes j'$.
  It follows that $H \circ (\id \otimes j'^\dag)$
  is the inverse of $M$.
\end{proof}
  
Since the category of completely positive maps is self-dual in a way that fixes the inclusion of \cat{FHilb}, the collection of copure morphisms is precisely the same as the collection of pure ones. This coincidence is special to the quantum case. The symmetric monoidal subcategory $\mathcal D$ consists precisely of the simple discarding maps. 
  
\begin{theorem}
  The category \cat{Quant} has purification and copurification. Hence the pairs $(\mathcal M, \mathcal P)$ and $(\mathcal C, \mathcal D)$ are symmetric monoidal weak factorisation systems.
\end{theorem}
\begin{proof}
  Purification follows from Lemma~\ref{lem:quantum:pure} and Theorem~\ref{thm:stinespring}, copurification by taking duals.
\end{proof}

\section{Related work}\label{sec:related}

Chiribella~\cite{chiribella:pure} defines a map $p$ to be pure if:
\[
    \begin{pic}[xscale=.75,yscale=.5]
      \node[morphism] (2) at ([yshift=20]1.north west) {$p$};
      \draw (2.south) to +(0,-0.7);
      \draw (2.north) to +(0,.7);
    \end{pic}
    =
    \begin{pic}[xscale=.75,yscale=.5]
      \node[morphism] (1) at (0,0) {$b$};
      \draw (1.north east) to +(0,0.3)node[ground,scale=0.7]{};
      \draw (1.north west) to +(0,.7);
      \draw (1.south) to +(0,-.7);
    \end{pic}
    \qquad \implies \qquad
    \text{ $b$ is of the form }\;
    \begin{pic}[xscale=.75,yscale=.5]
      \node[morphism] (1) at (0,-0) {$h$};
      \node[morphism] (2) at (-1,0) {$p$};
      \draw (2.south) to +(0,-.5);
      \draw (1.north) to +(0,.5);
      \draw (2.north) to +(0,.5);
    \end{pic}
\]
Selby and Coecke~\cite{selbycoecke:leaks} define a map $p$ to be pure if:
\[
    \begin{pic}[xscale=.75,yscale=.5]
      \node[morphism] (2) at ([yshift=20]1.north west) {$p$};
      \draw (2.south) to +(0,-0.7);
      \draw (2.north) to +(0,.7);
    \end{pic}
    =
    \begin{pic}[xscale=.75,yscale=.5]
      \node[morphism] (1) at (0,0) {$b$};
      \draw (1.north east) to +(0,0.3)node[ground,scale=0.7]{};
      \draw (1.north west) to +(0,.7);
      \draw (1.south) to +(0,-.7);
    \end{pic}
    \qquad \implies \qquad\exists h,h':\quad
    \begin{pic}[xscale=.75,yscale=.5]
      \node[morphism] (1) at (0,0) {$b$};
      \draw (1.north east) to +(0,0.7);
      \draw (1.north west) to +(0,.7);
      \draw (1.south) to +(0,-.7);
    \end{pic}
    =
    \begin{pic}[xscale=.75,yscale=.5]
      \node[morphism] (1) at (0,-0) {$h$};
      \node[morphism] (2) at ([yshift=-20]1.south) {$p$};
      \draw (1.south) to (2.north);
      \draw (1.north east) to +(0,.5);
      \draw (1.north west) to +(0,.5);
      \draw (2.south) to +(0,-.5);
    \end{pic}
    =
    \begin{pic}[xscale=.75,yscale=.5]
      \node[morphism] (2) at (0,0) {$h'$};
      \node[morphism] (1) at ([yshift=20,xshift=-3]1.north west) {$p$};
      \draw (1.south) to (2.north west);
      \draw (1.north) to +(0,.5);
      \draw (1.north) to +(0,.5);
      \draw (2.south) to +(0,-.5);
      \draw ([xshift=2]2.north east) to +(0,1.47);
    \end{pic}
    \quad\text{and}\quad
    \begin{pic}[xscale=.75,yscale=.5]
      \node[morphism] (1) at (0,0) {$h$};
      \draw (1.north east) to +(0,0.3)node[ground,scale=0.7]{};
      \draw (1.north west) to +(0,0.7);
      \draw (1.south) to +(0,-.7);
    \end{pic}
    =
    \begin{pic}[xscale=.75,yscale=.5]
      \draw (0,0) to +(0,2);
    \end{pic}
     \quad\text{and}\quad
    \begin{pic}[xscale=.75,yscale=.5]
      \node[morphism] (1) at (0,0) {$h'$};
      \draw (1.north east) to +(0,0.3)node[ground,scale=0.7]{};
      \draw (1.north west) to +(0,0.7);
      \draw (1.south) to +(0,-.7);
    \end{pic}
    =
    \begin{pic}[xscale=.75,yscale=.5]
      \draw (0,0) to +(0,2);
    \end{pic}
\]
Both definitions yield the same pure subcategory in \cat{Quant}, but they disagree in the other cases. For example, consider \cat{FRel}. According to $(*)$ the pure morphisms in \cat{FRel} are the partial functions. According to Selby and Coecke's definition the collection of pure morphisms is smaller; only the partial injections. The collection is smaller yet under Chiribella's definition; only the relations which relate at most one pair of elements. So these definitions produce a stricter notion of purity than that of this paper.

Another difference is that these definitions define purity in relation to the discarding map $\,\tinyground$ rather than the completely mixed state $\,\tinymix$. This means that the resulting collection of pure maps often bares some relation to the maps which we would call copure. This can also be seen in \cat{FRel}; the maps which are pure under Selby and Coecke's definition are precisely those in $\mathcal P\cap \mathcal C$. 

Finally, neither of these two definitions guarantee that the collection of pure maps form a monoidal subcategory. We have already seen this in the case of Chiribella's definition, which says that the identity morphisms are not pure in \cat{FRel}. To see this for Selby and Coecke's definition consider the category given by the usual order on the positive integers, with $\times$ as the monoidal product and the morphism $(n\geq 1)$ as \tinyground. Then a morphism $(n\geq m)$ is pure iff $m>n/2$. Such morphisms are not closed under $\circ$ or $\otimes$. Of course $\mathcal P$ as defined by $(*)$ is always a monoidal category, because it is of the form $\Rlift{\mathcal A}$.

\bibliographystyle{eptcs}
\bibliography{mixing}

\begin{thebibliography}{10}
\providecommand{\bibitemdeclare}[2]{}
\providecommand{\surnamestart}{}
\providecommand{\surnameend}{}
\providecommand{\urlprefix}{Available at }
\providecommand{\url}[1]{\texttt{#1}}
\providecommand{\href}[2]{\texttt{#2}}
\providecommand{\urlalt}[2]{\href{#1}{#2}}
\providecommand{\doi}[1]{doi:\urlalt{http://dx.doi.org/#1}{#1}}
\providecommand{\bibinfo}[2]{#2}

\bibitemdeclare{inproceedings}{abramskycoecke:categoricalsemantics}
\bibitem{abramskycoecke:categoricalsemantics}
\bibinfo{author}{S.~\surnamestart Abramsky\surnameend} \&
  \bibinfo{author}{B.~\surnamestart Coecke\surnameend} (\bibinfo{year}{2004}):
  \emph{\bibinfo{title}{A categorical semantics of quantum protocols}}.
\newblock In: {\sl \bibinfo{booktitle}{Logic in Computer Science}},
  \bibinfo{publisher}{IEEE Computer Society}, pp. \bibinfo{pages}{415--425},
  \doi{10.1109/LICS.2004.1}.

\bibitemdeclare{book}{adamekherrlichstrecker:joyofcats}
\bibitem{adamekherrlichstrecker:joyofcats}
\bibinfo{author}{J.~\surnamestart Adamek\surnameend},
  \bibinfo{author}{H.~\surnamestart Herrlich\surnameend} \&
  \bibinfo{author}{G.~\surnamestart Strecker\surnameend}
  (\bibinfo{year}{2006}): \emph{\bibinfo{title}{Abstract and concrete
  categories: the joy of cats}}.
\newblock {\sl \bibinfo{series}{Reprint}}~\bibinfo{volume}{17},
  \bibinfo{publisher}{Theory and Applications of Categories}.

\bibitemdeclare{inproceedings}{chiribella:dilation}
\bibitem{chiribella:dilation}
\bibinfo{author}{G.~\surnamestart Chiribella\surnameend}
  (\bibinfo{year}{2014}): \emph{\bibinfo{title}{Dilation of states and
  processes in operational-probabilistic theories}}.
\newblock In: {\sl \bibinfo{booktitle}{Quantum Physics and Logic}}, {\sl
  \bibinfo{series}{EPTCS}} \bibinfo{volume}{172}, pp. \bibinfo{pages}{1--14},
  \doi{10.4204/EPTCS.172.1}.

\bibitemdeclare{article}{chiribella:pure}
\bibitem{chiribella:pure}
\bibinfo{author}{G.~\surnamestart Chiribella\surnameend}
  (\bibinfo{year}{2014}): \emph{\bibinfo{title}{Distinguishability and
  copiability of programs in general process theories}}.
\newblock {\sl \bibinfo{journal}{International Journal of Software and
  Informatics}} \bibinfo{volume}{8}(\bibinfo{number}{3--4}), pp.
  \bibinfo{pages}{209--223}.
\newblock \urlprefix\url{http://arxiv.org/abs/1411.3035}.

\bibitemdeclare{article}{chiribelladarianoperinotti:purification}
\bibitem{chiribelladarianoperinotti:purification}
\bibinfo{author}{G.~\surnamestart Chiribella\surnameend},
  \bibinfo{author}{G.~M. \surnamestart {D'Ariano}\surnameend} \&
  \bibinfo{author}{P.~\surnamestart Perinotti\surnameend}
  (\bibinfo{year}{2010}): \emph{\bibinfo{title}{Probabilistic theories with
  purification}}.
\newblock {\sl \bibinfo{journal}{Physical Review A}} \bibinfo{volume}{81}, p.
  \bibinfo{pages}{062348}, \doi{10.1103/PhysRevA.81.062348}.

\bibitemdeclare{inproceedings}{chiribellascandolo:diagonal}
\bibitem{chiribellascandolo:diagonal}
\bibinfo{author}{G.~\surnamestart Chiribella\surnameend} \&
  \bibinfo{author}{C.~M. \surnamestart Scandolo\surnameend}
  (\bibinfo{year}{2015}): \emph{\bibinfo{title}{Operational axioms for
  diagonalizing states}}.
\newblock In: {\sl \bibinfo{booktitle}{Quantum Physics and Logic}}, {\sl
  \bibinfo{series}{EPTCS}} \bibinfo{volume}{195}, pp. \bibinfo{pages}{96--115},
  \doi{10.4204/EPTCS.195.8}.

\bibitemdeclare{article}{choi}
\bibitem{choi}
\bibinfo{author}{M.-D. \surnamestart Choi\surnameend} (\bibinfo{year}{1975}):
  \emph{\bibinfo{title}{Completely positive linear maps on complex matrices}}.
\newblock {\sl \bibinfo{journal}{Linear Algebra and its Applications}}
  \bibinfo{volume}{10}, pp. \bibinfo{pages}{285--290},
  \doi{10.1016/0024-3795(75)90075-0}.

\bibitemdeclare{article}{coeckeheunen:cp}
\bibitem{coeckeheunen:cp}
\bibinfo{author}{B.~\surnamestart Coecke\surnameend} \&
  \bibinfo{author}{C.~\surnamestart Heunen\surnameend} (\bibinfo{year}{2016}):
  \emph{\bibinfo{title}{Pictures of complete positivity in arbitrary
  dimension}}.
\newblock {\sl \bibinfo{journal}{Information and Computation}}
  \bibinfo{volume}{250}, pp. \bibinfo{pages}{50--58},
  \doi{10.1016/j.ic.2016.02.007}.

\bibitemdeclare{book}{coeckekissinger:cqm}
\bibitem{coeckekissinger:cqm}
\bibinfo{author}{B.~\surnamestart Coecke\surnameend} \&
  \bibinfo{author}{A.~\surnamestart Kissinger\surnameend}
  (\bibinfo{year}{2017}): \emph{\bibinfo{title}{Picturing quantum processes: a
  first course in quantum theory and diagrammatic reasoning}}.
\newblock \bibinfo{publisher}{Cambridge University Press},
  \doi{10.1017/9781316219317}.

\bibitemdeclare{article}{coeckeperdrix:environment}
\bibitem{coeckeperdrix:environment}
\bibinfo{author}{B.~\surnamestart Coecke\surnameend} \&
  \bibinfo{author}{S.~\surnamestart Perdrix\surnameend} (\bibinfo{year}{2010}):
  \emph{\bibinfo{title}{Environment and classical channels in categorical
  quantum mechanics}}.
\newblock {\sl \bibinfo{journal}{Logical Methods in Computer Science}}
  \bibinfo{volume}{8}(\bibinfo{number}{4}), p.~\bibinfo{pages}{14},
  \doi{10.2168/LMCS-8(4:14)2012}.

\bibitemdeclare{inproceedings}{cunninghamheunen:cpstar}
\bibitem{cunninghamheunen:cpstar}
\bibinfo{author}{O.~\surnamestart Cunningham\surnameend} \&
  \bibinfo{author}{C.~\surnamestart Heunen\surnameend} (\bibinfo{year}{2015}):
  \emph{\bibinfo{title}{Axiomatizing complete positivity}}.
\newblock In: {\sl \bibinfo{booktitle}{Quantum Physics and Logic}}, {\sl
  \bibinfo{series}{EPTCS}} \bibinfo{volume}{195}, pp.
  \bibinfo{pages}{148--157}, \doi{10.4204/EPTCS.195.11}.

\bibitemdeclare{article}{freydkelly:factorisation}
\bibitem{freydkelly:factorisation}
\bibinfo{author}{P.~\surnamestart Freyd\surnameend} \& \bibinfo{author}{G.~M.
  \surnamestart Kelly\surnameend} (\bibinfo{year}{1972}):
  \emph{\bibinfo{title}{Categories of continuous functors {I}}}.
\newblock {\sl \bibinfo{journal}{Journal of Pure and Applied Algebra}}
  \bibinfo{volume}{2}, pp. \bibinfo{pages}{169--191},
  \doi{10.1016/0022-4049(72)90001-1}.

\bibitemdeclare{book}{heunenvicary:cqm}
\bibitem{heunenvicary:cqm}
\bibinfo{author}{C.~\surnamestart Heunen\surnameend} \&
  \bibinfo{author}{J.~\surnamestart Vicary\surnameend} (\bibinfo{year}{2017}):
  \emph{\bibinfo{title}{Categories for quantum theory: an introduction}}.
\newblock \bibinfo{publisher}{Oxford University Press}.

\bibitemdeclare{article}{joyal:quasicategories}
\bibitem{joyal:quasicategories}
\bibinfo{author}{A.~\surnamestart Joyal\surnameend} (\bibinfo{year}{2008}):
  \emph{\bibinfo{title}{The theory of quasi-categories and its applications}}.
\newblock {\sl \bibinfo{journal}{Quaderns}} \bibinfo{volume}{45}, pp.
  \bibinfo{pages}{149--496}.
\newblock
  \urlprefix\url{http://mat.uab.cat/~kock/crm/hocat/advanced-course/Quadern45-2.pdf}.

\bibitemdeclare{article}{selbycoecke:leaks}
\bibitem{selbycoecke:leaks}
\bibinfo{author}{J.~\surnamestart Selby\surnameend} \&
  \bibinfo{author}{B.~\surnamestart Coecke\surnameend} (\bibinfo{year}{2017}):
  \emph{\bibinfo{title}{Leaks: quantum, classical, intermediate, and more}}.
\newblock {\sl \bibinfo{journal}{Entropy}}
  \bibinfo{volume}{19}(\bibinfo{number}{4}), p. \bibinfo{pages}{174},
  \doi{10.3390/e19040174}.

\bibitemdeclare{inproceedings}{selinger:cpm}
\bibitem{selinger:cpm}
\bibinfo{author}{P.~\surnamestart Selinger\surnameend} (\bibinfo{year}{2007}):
  \emph{\bibinfo{title}{Dagger compact closed categories and completely
  positive maps}}.
\newblock In: {\sl \bibinfo{booktitle}{Quantum Physics and Logic}}, {\sl
  \bibinfo{series}{ENTCS}} \bibinfo{volume}{170}, pp.
  \bibinfo{pages}{139--163}, \doi{10.1016/j.entcs.2006.12.018}.

\bibitemdeclare{inproceedings}{selinger:graphicallanguages}
\bibitem{selinger:graphicallanguages}
\bibinfo{author}{P.~\surnamestart Selinger\surnameend} (\bibinfo{year}{2009}):
  \emph{\bibinfo{title}{A survey of graphical languages for monoidal
  categories}}.
\newblock In: {\sl \bibinfo{booktitle}{New Structures for Physics}},
  \bibinfo{series}{Lecture Notes in Physics}, \bibinfo{publisher}{Springer},
  pp. \bibinfo{pages}{289--355},
  \doi{\detokenize{10.1007/978-3-642-12821-9_4}}.

\bibitemdeclare{article}{stinespring}
\bibitem{stinespring}
\bibinfo{author}{W.~F. \surnamestart Stinespring\surnameend}
  (\bibinfo{year}{1955}): \emph{\bibinfo{title}{Positive functions on
  {C}*-algebras}}.
\newblock {\sl \bibinfo{journal}{Proceedings of the American Mathematical
  Society}} \bibinfo{volume}{6}, pp. \bibinfo{pages}{211--216},
  \doi{10.2307/2032342}.

\bibitemdeclare{inproceedings}{westerbaan:paschke}
\bibitem{westerbaan:paschke}
\bibinfo{author}{A.~\surnamestart Westerbaan\surnameend} \&
  \bibinfo{author}{B.~\surnamestart Westerbaan\surnameend}
  (\bibinfo{year}{2016}): \emph{\bibinfo{title}{Paschke dilations}}.
\newblock In: {\sl \bibinfo{booktitle}{Quantum Physics and Logic}}, {\sl
  \bibinfo{series}{EPTCS}} \bibinfo{volume}{236}, pp.
  \bibinfo{pages}{229--244}, \doi{10.4204/EPTCS.236.15}.

\end{thebibliography}

\appendix
\section{Proofs}\label{sec:proofs}

\begin{proof}[Proof of Lemma~\ref{lem:probabilistic:system}.]
  Any matrix $f:I\to J$ of nonnegative reals factors as follows. The matrix represented by
  $\id[I] \otimes \tinymix_J$
  has exactly one nonzero entry in each row and at least one nonzero entry in each column. Let $p:I\times J\to J$ be the matrix with $((i,j),j')$-entry $\delta_{jj'}f_{ij}$; this has at most one nonzero entry in each column. Then $f$ is given by:
  \begin{equation}\label{eq:matrix}
    \begin{pic}[xscale=.75,yscale=.5]
      \node[morphism] (a) at (0,.5) {$\;p\;$};
      \draw (a.north) to +(0,.5)node[left]{$J$};
      \draw (a.south west) to +(0,-.9)node[left]{$I$};
      \draw (a.south east) to +(0,-.3)node[right]{$J$} to +(0,-.6)node[ground,scale=-.75]{};
    \end{pic}
  \end{equation}
  So $(\mathcal L,\mathcal R)$ is a factorisation system.
  It remains to prove that $m\oslash p$ for $m\in\mathcal L$, $p\in\mathcal R$.
  It suffices to show $m\lift p$, because if $p\in\mathcal L$ then $p\times \mathrm{id}_A$ is also in $\mathcal L$, and similarly for $m\in\mathcal R$. 
  If $b \circ m = p \circ a$, 
  then we can factor $a$ and $b$ and look for $h$ making the diagram below left commute.
  So it suffices to prove that any commuting square on the right below has such an $h$. 
  \[
   \begin{tikzcd}
    \bullet\arrow{r}{m_a}\arrow{d}[swap]{m}&\bullet\arrow{r}{p_a}&\bullet\arrow{d}{p}\\
    \bullet\arrow{r}[swap]{m_b}&\bullet\arrow{u}[swap]{h}\arrow{r}[swap]{p_b}&\bullet
   \end{tikzcd}
   \qquad\qquad
   \begin{tikzcd}
    W\arrow{r}{m'}\arrow{d}[swap]{m}&Y\arrow{d}{p'}\\
    X\arrow{r}[swap]{p}\arrow[dashed]{ur}{h}&Z
   \end{tikzcd}
  \]
  Define $h$ as follows. For each $y\in Y$ with $p'_{zy}=0$ for all $z\in Z$, take the unique $w\in W$ with $m'_{yw}\neq 0$ and pick $x\in X$ with $m_{xw}\neq 0$, then set $h_{yx}=m'_{yw}/m_{xw}$. Now for each $w\in W$ and $z\in Z$ say $x\in X_{wz}\subseteq X$ if $m_{xw}\neq 0$ and $p_{ zx}\neq 0$, and similarly $y\in Y_{wz}\subseteq Y$ if $m'_{yw}\neq 0$ and $p'_{ zy}\neq 0$. Then since $pm=p'm'$,
  \[
    \sum_{x\in X_{xz}}p_{zx}m_{xw}=  \sum_{y\in Y_{xz}}p'_{zx}m'_{xw}\text.
  \]
  Call this quantity $c_{wz}$, and for each $x\in X_{wz}$, $y\in Y_{wz}$ define $h_{yx}=p_{zx}m'_{yw}/c_{wz}$. Let $h_{yx}=0$ if its value has not already been defined. Then $hm=m'$ and $p'h=p$, as required.
\end{proof}

\begin{proof}[Proof of Lemma~\ref{lem:possibilistic:system}.]
Any relation $r\colon A \to B$ factors as the introduction of a completely mixed state on $B$, given by the relation $A \to A \times B$ that relates $a \in A$ to $(a,b)$ for any $b \in B$, followed by a partial function $A \times B \to B$ that sends $(a,b)$ to $b$ when $(a,b) \in r$. Since the introduction of a completely mixed state is injective, surjective and total, we have that $(\mathcal L,\mathcal R)$ is a factorisation system.

It remains to prove that $m\oslash p$ for $m\in\mathcal L$, $p\in\mathcal R$. By the same logic as used in the proof of Lemma~\ref{lem:probabilistic:system} it suffices to prove that any commuting square
  \[
\begin{tikzcd}
  W\arrow{r}{m'}\arrow{d}[swap]{m}&Y\arrow{d}{p'}\\
  X\arrow{r}[swap]{p}\arrow[dashed]{ur}{h}&Z
\end{tikzcd}
\]
  has such an $h$. Let $h$ be the relation containing $(x,y)\in X\times Y$ if both $x$ and $y$ are related to the same $w\in W$ and either $x$ and $y$ are related to the same $z\in Z$ or $y$ is not related to any $z\in Z$.
\end{proof}

\begin{proof}[Proof of Lemma~\ref{lem:minimalstinespring}.]
  Parts~\ref{lem:minimal:exists} and~\ref{lem:minimal:epi} are well-known~\cite{westerbaan:paschke}. To prove~\ref{lem:extenddilation} it is easier to work in the dual case. Suppose $(\id \otimes \tinyground) \circ p = (\id \otimes \tinyground) \circ p'$.
 Then~\ref{lem:minimal:exists} gives a minimal Stinespring dilation $\tilde{p}$ with ancilla $\tilde{C}$ and isometries $i$ and $i'$ such that $p = (\id \otimes i) \circ \tilde{p}$ and $p' = (\id \otimes i') \circ \tilde{p}$.
 It suffices to exhibit an isometry $i'' \colon C\to C'$ with $i''\circ i=i'$. To construct $i''$, note that the space $C$ can be written as $\mathrm{im}(i)\oplus\mathrm{im}(i)^\perp$ and the space $C'$ can be written as $\mathrm{im}(i)\oplus\mathrm{im}(i')^\perp$. Since $\mathrm{im}(i)$ and $\mathrm{im}(i')$ are isometric copies of $\tilde{C}$, there is a canonical unitary $u$ between them. Pick an isometry $\mathrm{im}(i)^\perp \to \mathrm{im}(i')^\perp$ (which exists since $\mathrm{dim}(\mathrm{im}(i')^\perp)\geq\mathrm{dim}(\mathrm{im}(i)^\perp)$) and let $i''$ be the direct sum of $u$ and this isometry.
\end{proof}
\end{document}